\documentclass[a4paper,twocolumn,superscriptaddress,10pt]{quantumarticle}
\usepackage{graphicx}
\usepackage[pdftex,colorlinks=true,linkcolor=blue,citecolor=blue,urlcolor=blue]{hyperref}
\usepackage{amsmath, amsthm, amssymb}
\usepackage{subfigure}
\usepackage{comment}
\usepackage{url}
\usepackage[ruled,lined,linesnumbered]{algorithm2e}
\usepackage{tikz}
\usepackage{calc}
\usepackage{multirow}

\newcommand{\nc}{\newcommand}
\nc{\rnc}{\renewcommand}
\nc\benum{\begin{enumerate}}
\nc\eenum{\end{enumerate}}
\nc\bit{\begin{itemize}}
\nc\eit{\end{itemize}}

\newcommand{\secref}[1]{Section~\ref{sec:#1}}

\newcommand{\thmref}[1]{Theorem~\ref{thm:#1}}

\newcommand{\R}{\mathbb{R}}

\newcommand{\E}{\mathbb{E}}
\newcommand{\F}{\mathbb{F}}

\newcommand{\ket}[1]{| #1 \rangle}
\newcommand{\bra}[1]{\langle #1|}

\newcommand{\bracket}[3]{\langle #1|#2|#3 \rangle}

\DeclareMathOperator{\poly}{poly}

\DeclareMathOperator{\tr}{tr}

\DeclareMathOperator{\Var}{Var}
\DeclareMathOperator{\Inf}{Inf}
\DeclareMathOperator{\sgn}{sgn}

\newcommand{\be}{\begin{equation}}
\newcommand{\ee}{\end{equation}}
\newcommand{\bea}{\begin{eqnarray}}
\newcommand{\eea}{\end{eqnarray}}
\newcommand{\bes}{\begin{equation*}}
\newcommand{\ees}{\end{equation*}}
\newcommand{\beas}{\begin{eqnarray*}}
\newcommand{\eeas}{\end{eqnarray*}}

\makeatletter
\newtheorem*{rep@theorem}{\rep@title}
\newcommand{\newreptheorem}[2]{%
\newenvironment{rep#1}[1]{%
 \def\rep@title{#2 \ref{##1} (restated)}%
 \begin{rep@theorem}}%
 {\end{rep@theorem}}}
\makeatother

\newtheorem{thm}{Theorem}
\newtheorem*{thm*}{Theorem}

\newtheorem{lem}[thm]{Lemma}

\newtheorem*{lem*}{Lemma}

\newreptheorem{thm}{Theorem}
\newreptheorem{lem}{Lemma}

\def\eq#1{(\ref{eq:#1})}

\def\begsub#1#2\endsub{\begin{subequations}\label{eq:#1}\begin{align}#2\end{align}\end{subequations}}
\nc\qand{\qquad\text{and}\qquad}
\nc\mnb[1]{\medskip\noindent{\bf #1}}

\begin{document}

\title{Extremal eigenvalues of local Hamiltonians}

\author{Aram W. Harrow}
\affiliation{Center for Theoretical Physics, Department of Physics,
  MIT}
\email{aram@mit.edu}
\author{Ashley Montanaro}
\affiliation{School of Mathematics, University of Bristol, UK}
\email{ashley.montanaro@bristol.ac.uk}

\begin{abstract}
We apply classical algorithms for approximately solving constraint
satisfaction problems to find bounds on extremal eigenvalues of local
Hamiltonians. We consider spin Hamiltonians for which we have an upper
bound on the number of terms in which each spin participates, and find
extensive bounds for the operator norm and ground-state
energy of such Hamiltonians under this constraint.  
In each case the
bound is achieved by a product state which can be found efficiently
using a classical algorithm. 
\end{abstract}

\maketitle


\section{Introduction}

The eigenvalue statistics of a local Hamiltonian are related to its
structure. One example is the level spacings of chaotic vs integrable systems,
which can be seen as the small-scale structure of the spectrum.  What
about the large-scale features, such as the extremal eigenvalues?  Do
these scale differently for non-interacting or interacting systems?
It is generally understood that interacting systems can be frustrated,
meaning that all the local terms cannot simultaneously be in their
ground state.  This situation is generically true for local terms with
entangled ground states. But how much of an effect can frustration have on
the ground-state energy of a system? 

Here we study the extremal eigenvalues of quantum Hamiltonians which are only weakly interacting, in the sense that they can be
written as sums of terms where each term depends only on a few qubits,
and each qubit is included in only a few terms.  With this mild form
of locality imposed, how far apart must the largest and smallest
eigenvalues be?  If the Hamiltonian were non-interacting, the
separation should scale with the size of the system.  For a more
general Hamiltonian, the extremal eigenvectors may be highly entangled
and interacting terms may contribute opposite signs.  Nevertheless, in
this paper we show lower bounds on the norms of local Hamiltonians
under very general conditions.  An additional argument shows
specifically that the ground-state energy is low (or if desired, that
the top eigenvalue is high).

\begin{thm}
\label{thm:main}
Let $H$ be a traceless $k$-local Hamiltonian on $n$ qubits such that
$k=O(1)$. Assume that $H$ can be expressed as a weighted sum of $m$
distinct Pauli terms such that each term has operator norm
$\Theta(1)$, and each qubit participates in at most $\ell$ terms. Then
$\|H\| \ge \Omega(m / \sqrt{\ell} )$ and $\lambda_{\min}(H) \le -
\Omega(m / \ell)$. In each case the bound is achieved by a product
state which can be found efficiently using a classical algorithm. 
\end{thm}

In the above theorem, $\|H\|$ is the operator norm of $H$ and
$\lambda_{\min}(H)$ is the lowest eigenvalue (ground-state energy) of
$H$.  (Of course a similar statement could also be made about
$\lambda_{\max}$.  We focus on $\lambda_{\min}$ because of its
relevance to physical systems and to constraint satisfaction
problems.)  The notation $O(f(x)), \Omega(f(x))$ refers to functions
that are $\leq c f(x)$ or $\geq cf(x)$, respectively, for some
absolute constant $c$. We write $\Theta(f(x))$ to mean both a function
that is both $O(f(x))$ and $\Omega(f(x))$.

If $H$ were a non-interacting Hamiltonian ($k=\ell=1$) the largest and smallest
eigenvalues would both be $\Theta(m)=\Theta(n)$.  Thus \thmref{main} can be
viewed as saying that interaction can reduce the norm of $H$ by at
most a $O(\sqrt{\ell})$ factor and can reduce the smallest eigenvalue
by at most a $O(\ell)$ factor. Observe that the bound on
$\lambda_{\min}(H)$ is $-\Omega(n)$ for all lattice Hamiltonians. This
proves that for any such system the ground-state energy is smaller
than the average energy by an extensive amount.  By constrast, using
our information about $\tr H^2$ alone would only show that
$\lambda_{\min}(H) \leq -\Omega(\sqrt{m})$, which is in general a
vanishing fraction of system size.

The restriction to terms of weight $\Theta(1)$ in Theorem \ref{thm:main} is not essential and is only included to simplify the bounds. Further, the hidden constants are not overly large for small $k$; more precise statements of our results are given below. For example, for 2-local qubit Hamiltonians, the precise bound on $\lambda_{\min}$ we obtain is $\lambda_{\min}(H) \le - \|\widehat{H}\|_1 / (24\ell)$, where $\|\widehat{H}\|_1$ is the sum of the absolute values of the coefficients in the Pauli expansion of $H$. As a simple instance where this bound can be applied, consider the antiferromagnetic Heisenberg model $H = \sum_{\langle i,j \rangle} X_i X_j + Y_i Y_j + Z_i Z_j$ on a regular lattice with $n$ vertices. Then, for {\em any} such lattice, we obtain $\lambda_{\min}(H)/n \le -1/48$. 

Theorem \ref{thm:main} can be applied to qudit Hamiltonians with local dimension $d > 2$ by embedding each subsystem in $\lceil \log_2 d\rceil$ qubits at the expense of increasing the locality from $k$ to $k \lceil \log_2 d \rceil$.

{\bf Proof outline.}
Both results that make up Theorem \ref{thm:main} are based on the use
of a correspondence between local quantum Hamiltonians and low-degree
polynomials, which allows us to apply classical approximation
algorithms for constraint satisfaction problems. This correspondence
uses a qubit 2-design~\cite{delsarte77,renes04} to convert arbitrary
qubit Hamiltonians to polynomials over boolean variables.

The operator norm bound in Theorem \ref{thm:main} (stated more precisely as Lemma \ref{lem:quantum1} below) is based on recent work of Barak et al.~\cite{barak15} which gives an efficient randomised algorithm for satisfying a relatively large fraction of a set of linear equations over $\F_2$. The bound on $\lambda_{\min}$ (stated more precisely as Lemma \ref{lem:quantum3} below) is based on analysing a natural greedy algorithm which is similar to a classical algorithm of H\aa stad~\cite{hastad00}. Our results can be seen as generalising these two classical algorithms to the quantum regime.

{\bf Other related work.} Bansal, Bravyi and
Terhal~\cite{bansal09} have previously shown that, for 2-local qubit
Hamiltonians $H$ on a planar graph with Pauli interactions of weight
$\Theta(1)$, $\lambda_{\min}(H) \le -\Omega(m)$. Similarly to our
result, their proof uses a mapping between quantum and classical
Hamiltonians and proves the existence of a product state achieving a
$-\Omega(m)$ bound. However, the two results are not comparable; ours
holds for non-planar graphs and $k$-local Hamiltonians for $k>2$,
while theirs encompasses two-local Hamiltonians on planar graphs with
vertices of arbitrarily high degree. The quantum-classical mapping
used is also different.  Finally, the constants in our results are somewhat better
(for example, they obtain $\lambda_{\min}(H)/n \le -1/135$ for the antiferromagnetic Heisenberg model on a 2D triangular lattice).

This work was motivated by \cite{barak15} (whose main result is
presented in \secref{operator}).  Ref.~\cite{barak15} in turn was
inspired by~\cite{QAOA,QAOA2}, which gives a quantum algorithm for
finding low-energy states of classical Hamiltonians.  
The relative performance of these different algorithms
(ours/\cite{barak15} vs.~\cite{QAOA,QAOA2}) is in general unknown,
and it is also open to determine the extent to which~\cite{QAOA} can
be generalised to finding low-energy states of local Hamiltonians.

One other related work is~\cite{BH-product}, which showed that when
$k=2$ and the degree of the interaction graph is large, then product states can provide a good
approximation for any state, with respect to the metric given by
averaging the trace distance over the pairs of systems acted on by the
Hamiltonian.  In particular this means they can approximate the
minimum and maximum eigenvalues.  Both our result and
\cite{BH-product} yield similar error bounds (ours are somewhat
tighter), but in this sense apply to incomparable settings:
\cite{BH-product} show that product states nearly match the energy of
some other state (e.g.~the true ground state) with possibly unknown
energy while our paper puts explicit bounds on the maximum and/or
minimum energy.

Another way to think about our work is as showing that interacting
spins must nevertheless behave in some ways like noninteracting spins.
In this picture, some vaguely related work is \cite{HMH04,BrandaoC15},
which show that under some conditions lattice systems have a density
of states that is approximately Gaussian.  These results are
incomparable to ours, even aside from the different assumptions,
because we put bounds on the extremal eigenvalues while they study the
density of states and/or the thermal states at nonzero temperature.
Theorem 4 of Ref.~\cite{Montanaro12} also bounded the density of states of
$k$-local Hamiltonians under general conditions, but in the opposite
direction: i.e.~putting upper bounds on how many eigenvalues could have
large absolute value.

{\bf Why product states?} Ground states of local Hamiltonians may be
highly entangled~\cite{FreedmanH13}.
But our
bounds on $\|H\|$ and $\lambda_{\min}(H)$ are achieved only with
product states. One reason for this in the case of $\|H\|$ is that we are using {\em random} states,
and product states have much larger
fluctuations than generic entangled states.  Indeed the variance of
$\bra\psi H \ket \psi$ for a random unit vector $\ket\psi$ is only
$O(m/2^n)$.  It is an interesting open question to find a distribution
over entangled states that improves the constant factors in
\thmref{main} that we achieve with product states.

{\bf Fourier analysis of boolean functions.} We will need some basic facts from classical Fourier analysis of boolean functions~\cite{odonnell14}. Any function $f:\{\pm1\}^n \rightarrow \R$ can be written as
\be f(x) = \sum_{S \subseteq [n]} \hat{f}(S) x_S,
\label{eq:f-hat}\ee
where $x_S := \prod_{i \in S} x_i$ and $[n] := \{1,\dots,n\}$. This is known as the Fourier expansion of $f$. Parseval's equality implies that
\be \label{eq:var} \Var(f) := \E_x[f(x)^2] - \E_x[f(x)]^2 = \sum_{S \neq \emptyset} \hat{f}(S)^2, \ee
where the expectation is taken over the uniform distribution on $\{\pm1\}^n$. In addition, $\hat{f}(\emptyset) = \E_x[f(x)]$. The influence of the $j$'th coordinate on $f$ is defined as
\[ \Inf_j(f) = \sum_{S \ni j} \hat{f}(S)^2. \]


\section{The quantum-classical correspondence}
\label{sec:qclass}

Let $H$ be a $k$-local Hamiltonian which has Pauli expansion
\[ H = \sum_{s \in \{I,X,Y,Z\}^n} \hat H_s\, s_1 \otimes s_2 \otimes \dots \otimes s_n \]
for some weights $\widehat{H}_s$ that we can view as a Fourier expansion of
$H$ analogous to that in \eq{f-hat}.  Define the norms
$\|\widehat{H}\|_p := (\sum_s |\widehat{H}_s|^p)^{1/p}$ for $p \ge
1$. In order to apply classical bounds to extremal eigenvalues of $H$,
we observe that the action of a $k$-local Hamiltonian on product
states corresponds to a low-degree polynomial. Define the following
set of states~\cite{renes04,rehacek04}: 
\begin{multline*} \ket{\psi_{++}} = \frac{1}{\sqrt{6}}\left( \sqrt{3+\sqrt{3}}\ket{0} + e^{i \pi/4} \sqrt{3-\sqrt{3}} \ket{1}\right),\\ \ket{\psi_{-+}} = Z \ket{\psi_{++}},\; \ket{\psi_{+-}} = X \ket{\psi_{++}},\; \ket{\psi_{--}} = Y \ket{\psi_{++}}.
\end{multline*}
%
These four states form a qubit 2-design; equivalently, a symmetric informationally-complete quantum measurement (SIC-POVM) on one qubit~\cite{renes04}. This measurement was studied in detail in~\cite{rehacek04}. Geometrically, the states describe a tetrahedron within the Bloch sphere~\cite{delsarte77}.

Then define the functions $\chi_s:\{\pm1\}^2 \rightarrow \R$, for $s \in \{I,X,Y,Z\}$, by
\[ \chi_s(x) = \bracket{\psi_x}{s}{\psi_x}. \]
These functions are pleasingly simple: one can verify that
\be \label{eq:funcs} \chi_I(x) = 1,\; \chi_X(x) = \frac{x_1}{\sqrt{3}},
\chi_Y(x) = \frac{x_1x_2}{\sqrt{3}},\;\chi_Z(x) = \frac{x_2}{\sqrt{3}}. \ee
%
Split each $x \in \{\pm1\}^{2n}$ into $n$ consecutive blocks of length 2, written as $x = x^{(1)} x^{(2)} \dots x^{(n)}$, and define the function $f_H:\{\pm1\}^{2n} \rightarrow \R$ by
\beas f_H(x) &=& \bra{\psi_{x^{(1)}}}\dots \bracket{\psi_{x^{(n)}}}{H}{\psi_{x^{(1)}}}\dots\ket {\psi_{x^{(n)}}}\\
&=& \sum_{s \in \{I,X,Y,Z\}^n} \widehat{H}_s \chi_{s_1}(x^{(1)}) \chi_{s_2}(x^{(2)}) \dots \chi_{s_n}(x^{(n)}).
\eeas
As each $x \in \{\pm1\}^{2n}$ corresponds to a state
$\ket{\psi_{x^{(1)}}}\dots\ket {\psi_{x^{(n)}}}$, we immediately have the bounds 
$\lambda_{\max}(H) \ge \max_{x \in \{\pm1\}^{2n}} f_H(x)$ and
$\lambda_{\min}(H) \le \min_{x \in \{\pm1\}^{2n}} f_H(x)$.   We will now
proceed to show bounds on $\max_{x \in \{\pm1\}^{2n}} f_H(x)$ and
$\min_{x \in \{\pm1\}^{2n}} f_H(x)$ by viewing $f_H(x)$ as a polynomial.
 
As $H$ is $k$-local and each function $\chi_s$ ($s \neq I$) is a monomial of
degree at most 2, $f_H$ is a polynomial of degree at most $2k$. 
Because the Fourier expansion of each function $\chi_s$ contains only one term, each term in $H$ corresponds to exactly one term in the Fourier expansion of $f_H$. Indeed
\be
\widehat{f_H}(s) = \widehat{H}_s 3^{-|s|/2},
\label{eq:fH-Fourier}\ee
where $s \in \{I,X,Y,Z\}^n$ and $|s| = |\{i: s_i \neq I\}|$. This
corresponds to identifying $\{I,X,Y,Z\}^n$ with subsets of $[2n]$ in
the natural way.
Thus by eqns.\ (\ref{eq:var}), (\ref{eq:funcs}) and \eq{fH-Fourier}
we have
\begin{align} \label{eq:varfh} \Var(f_H) &= \sum_{s \in \{I,X,Y,Z\}^n, s \neq I^n}
                                           \widehat{H}_s^2\,3^{-|s|}
\\
 \Inf_j(f_H) &= \sum_{s,s_j \neq I} \widehat{H}_s^2\,3^{-|s|}.
\nonumber\end{align}

%
%
%
%
%
%


\section{Operator norm bounds}\label{sec:operator}

We will use the following result of Barak et al.~\cite{barak15}, which is a constructive version of a probabilistic bound previously shown by Dinur et al.~\cite{dinur07}:

\begin{thm}[Barak et al.~\cite{barak15}]
\label{thm:classical}
There is a universal constant $C$ and a randomised algorithm such that
the following holds. Let $f:\{\pm1\}^n \rightarrow \R$ be a polynomial
with degree at most $k$ such that $\Var(f) = 1$. Let $t \ge 1$ and
suppose that $\Inf_i(f) \le C^{-k} t^{-2}$ for all $i \in [n]$. Then
with high probability the algorithm outputs $x \in \{\pm1\}^n$ such
that $|f(x)| \ge t$. The algorithm runs in time $\poly(m,n,\exp(k))$, where $m$ is the number of nonzero monomials in $f$.
\end{thm}

Recent independent work of H\aa stad~\cite{hastad15} describes an alternative, randomised algorithm achieving a similar bound.

Given the quantum-classical correspondence discussed in the previous section, we can now apply Theorem \ref{thm:classical} to $f_H$ to prove the following result, which is one half of Theorem \ref{thm:main}.

\begin{lem}
\label{lem:quantum1}
There is a universal constant $D$ and a randomised classical algorithm such that the following holds. Let $H$ be a traceless $k$-local Hamiltonian given as a weighted sum of $m$ Pauli terms such that, for all $j$, $\Inf_j(f_H) \le I_{\max}$. Then with high probability the algorithm outputs a product state $\ket{\psi}$ such that $|\bracket{\psi}{H}{\psi}| \ge D^{-k} \|\widehat{H}\|_2^2 / \sqrt{I_{\max}}$. The running time of the algorithm is $\poly(m,n,\exp(k))$.
\end{lem}

\begin{proof}
First observe that if we simply pick $x \in \{\pm1\}^{2n}$ uniformly at random and consider the corresponding product state $\ket{\psi_x}$,
\[ \E_x[ \bracket{\psi_x}{H}{\psi_x}^2] = \Var(f_H) \ge \frac{\|\widehat{H}\|_2^2}{3^k} \]
by (\ref{eq:varfh}). In addition (see e.g.~\cite[Theorem 9.24]{odonnell14}), as $f_H$ is a degree-$2k$ polynomial,
\[ \Pr_x[|f_H(x)| \ge \sqrt{\Var(f_H)}] \ge \exp(-O(k)). \]
Therefore, simply picking $\exp(O(k))$ random product states of the form
$\ket{\psi_x}$ achieves $|\bracket{\psi_x}{H}{\psi_x}|
\ge \|\widehat{H}\|_2/3^{k/2}$ with high probability.
Let $E$ be a universal constant to be chosen later. If $I_{\max} \ge E^{-k}\|\widehat{H}\|_2^2$, then $|\bracket{\psi_x}{H}{\psi_x}| \ge (\sqrt{3E})^{-k} \|\widehat{H}\|_2^2 / \sqrt{I_{\max}}$ as
desired, taking $D = \sqrt{3E}$. So assume henceforth that $I_{\max} \le E^{-k} \|\widehat{H}\|_2^2$. Let $f'_H =
f_H / \sqrt{\Var(f_H)}$ so $\Var(f'_H) = 1$. Then 
\[ \Inf_j(f'_H) =
\frac{\Inf_j(f_H)}{ \Var(f_H)} \le
3^k \frac{I_{\max}}{ \|\widehat{H}\|_2^2 }
.\]
 Set 
\[ t = \frac{C^{-k/2}}{\sqrt{\max_i \Inf_i(f'_H)}} \ge
 E^{-k/2} \frac{ \|\widehat{H}\|_2 }{\sqrt{I_{\max}}} \ge 1,
\]
where $C$ is the constant in Theorem \ref{thm:classical} and we choose
$E$ large enough for the first inequality to hold. Then the algorithm
of Theorem \ref{thm:classical} outputs $x \in \{\pm1\}^{2n}$ such that
$|f'_H(x)| \ge E^{-k/2} \|\widehat{H}\|_2/ \sqrt{I_{\max}}
$. Renormalising again by multiplying by $\sqrt{\Var(f_H)} \ge
3^{-k/2} \|\widehat{H}\|_2$, $|\bracket{\psi_x}{H}{\psi_x}| = |f_H(x)|
\ge (\sqrt{3E})^{-k} \|\widehat{H}\|_2^2/\sqrt{I_{\max}}$, which
completes the proof.

Note that the algorithm does not need to know whether $I_{\max}$ is
large or not, since it can simply try both strategies and see which
one results in the larger value of $|\bra\psi H \ket\psi|$.
\end{proof}

The first part of Theorem \ref{thm:main} is now immediate from Lemma \ref{lem:quantum1}. Let $H$ be a $k$-local Hamiltonian with $k=O(1)$ such that $H$ is a sum of $m$ distinct Pauli terms, each of weight $\Theta(1)$, with each qubit participating in $\ell$ terms. Then $\|\widehat{H}\|_2^2 = \Theta(m)$, $\Inf_j(f_H) = O(\ell)$.


\section{Bounds on extremal eigenvalues}

We now describe an algorithm for bounding extremal eigenvalues which is weaker, but holds for both the largest and smallest eigenvalues. Once again, the algorithm is based on applying the quantum-classical correspondence in Section \ref{sec:qclass} to a classical algorithm. We first describe the classical algorithm, which is a simple greedy approach to find large values taken by a low-degree polynomial on the boolean cube.

Given $f:\{\pm1\}^n \rightarrow \R$ define $W:= W(f) = \sum_{S \neq
  \emptyset} |\hat{f}(S)| $.
Assume that, for all $i \in [n]$,
\[ |\{T \subseteq [n]: \hat{f}(T) \neq 0 \text{ and } i \in T\}| \le \ell. \]
Consider the following algorithm, based on ideas of~\cite{hastad00} but somewhat simpler:
\begin{enumerate}
\item Find $S$ such that $|\hat{f}(S)|$ is maximal.
\item Substitute values for $x_i$, $i \in S$, such that
  $\hat{f}(\emptyset)$ increases by at least $|\hat{f}(S)|$.
\item Repeat until $f$ is constant; call this constant $f_{\text{end}}$.
\end{enumerate}
It is not obvious that step (2) can be achieved, because there might exist $T \subsetneq S$ such that $\hat{f}(T) \neq 0$. Define a function $f_S$ by $f_S(x) = \sum_{T \subseteq S} \hat{f}(T) x_T$. For each $T \subsetneq S$ such that $T \neq \emptyset$ and each $a \in \{\pm1\}$, $\E_{x,x_S = a}[x_T] = 0$. So $\E_{x,x_S=\sgn(\hat{f}(S))}[f_S(x)] = \hat{f}(\emptyset) + |\hat{f}(S)|$, and there must exist some $y$ achieving $f_S(y) \ge \hat{f}(\emptyset) + |\hat{f}(S)|$. Searching over at most $2^k$ different values $x$ is sufficient to find $y$.

\begin{lem}
\label{lem:quantum2}
When the above algorithm terminates, $f_{\text{end}} \ge \hat{f}(\emptyset) + W / (2k \ell)$.
\end{lem}

\begin{proof}
Let $f_j$ be the new function produced at the $j$'th stage of the
algorithm, with $f_0 = f$. Let $M_j$ be the value of
$|\widehat{f_{j-1}}(S)|$ corresponding to the set $S$ chosen at stage
$j$. Then $\widehat{f_j}(\emptyset) = \widehat{f_{j-1}}(\emptyset) +
M_j$ and $W(f_j) \ge W(f_{j-1}) - 2k \ell M_j$. The latter inequality
is shown as follows. For each $i \in S$, there are at most $\ell$
subsets $T$ such that $\widehat{f_{j-1}}(T) \neq 0$ and $i \in T$. So
there are at most $k\ell$ subsets $T$ such that $T \cap S \neq
\emptyset$. For each such $T$, the substitution of values $x_i$, $i
\in S$, implies that $\widehat{f_j}(T)$ is set to 0, and for some
other subset $T'$, $\widehat{f_j}(T') = \widehat{f_{j-1}}(T') \pm
\widehat{f_{j-1}}(T)$. These are the only coefficients modified by the
substitution process. Thus $W(f_j)$ can only decrease by at most
$2|\widehat{f_{j-1}}(T)| \le 2M_j$ for each $T$ such that $S \cap T
\neq \emptyset$. 
Rearranging we have $M_j \ge (W(f_{j-1}) - W(f_j))/(2k \ell)$ for $j \ge 1$ and hence
\beas f_{\text{end}} &=& \hat{f}(\emptyset) + \sum_j M_j\\
&\ge& \hat{f}(\emptyset) + \frac{1}{2k\ell} \sum_j W(f_{j-1}) - W(f_j)\\
&=& \hat{f}(\emptyset) + \frac{W}{2k\ell}
\eeas
as claimed.
\end{proof}

The following lemma is now essentially immediate.

\begin{lem}
\label{lem:quantum3}
There is a universal constant $E$ and a deterministic classical
algorithm such that the following holds. Let $H$ be a traceless
$k$-local Hamiltonian which can be written as a weighted sum of $m$
distinct Pauli terms such that each qubit participates in at most $\ell$ terms. Then
the algorithm outputs a product state $\ket{\psi}$ such that
$\bracket{\psi}{H}{\psi} \ge E^{-k} \|\widehat{H}\|_1 / \ell$. The algorithm runs in time
$\poly(m,n,\exp(k))$.
\end{lem}

\begin{proof}
Apply the algorithm of Lemma \ref{lem:quantum2} to the degree-$2k$ polynomial $f_H:\{\pm1\}^{2n} \rightarrow \R$ defined as in Section \ref{sec:qclass}. We have $W(f_H) \ge 3^{-k/2} \|\widehat{H}\|_1$, $\widehat{f_H}(\emptyset) = 0$. Hence the algorithm finds $\ket{\psi}$ such that $\bracket{\psi}{H}{\psi} \ge 3^{-k/2} \|\widehat{H}\|_1 / (4k \ell) \ge E^{-k} \|\widehat{H}\|_1 / \ell$ for a large enough universal constant $E$.
\end{proof}

Applying the same procedure to $-H$ is of course sufficient to also find $\ket{\psi}$ such that $\bracket{\psi}{H}{\psi} \le -E^{-k} \|\widehat{H}\|_1 / \ell$. Observing that $\|\widehat{H}\|_1 = \Theta(m)$ for a Hamiltonian $H$ which is a sum of $m$ distinct Pauli terms of weight $\Theta(1)$ completes the proof of the second part of Theorem \ref{thm:main}.

\section{Optimality}
Both bounds in Theorem \ref{thm:main} are tight, even for classical Hamiltonians, as demonstrated by the following examples. The first is based on an example in~\cite{barak15}. Consider the 2-local Hamiltonian on $n$ qubits
\[ H = \sum_{(i,j) \in E} \widehat{H}_{ij} \, Z_i Z_j, \]
where $E$ is the edges of an arbitrary $r$-regular undirected graph on $n$ vertices, and each weight $\widehat{H}_{ij} \in \{\pm1\}$ is picked uniformly at random. Then $m = rn/2$, $\ell = r$. For each fixed $x \in \{0,1\}^n$,
\[ \bracket{x}{H}{x} = \sum_{(i,j) \in E} \widehat{H}_{ij} (-1)^{x_i + x_j} \]
is a sum of $rn/2$ uniformly distributed elements of $\{\pm1\}$. Via a standard Chernoff bound argument,
\[ \Pr_H[ |\bracket{x}{H}{x}| \ge t] \le 2e^{-\frac{t^2}{rn}}. \]
Fixing $t = \Theta(n\sqrt{r})$ and taking a union bound over all $x \in \{0,1\}^n$, $\|H\| = O(n\sqrt{r}) = O(m/\sqrt{\ell})$ with high probability.

Second, consider the 2-local Hamiltonian on $n$ qubits
\[ H = \sum_{i < j} Z_i Z_j. \]
Then $H$ is a sum of $m = \Theta(n^2)$ Pauli terms of weight 1, where each qubit participates in $\ell = \Theta(n)$ terms. We have
\[ \bracket{x}{H}{x} = \sum_{i<j} (-1)^{x_i + x_j} = \frac{(n-2|x|)^2-n}{2}, \]
%
so $\lambda_{\min}(H) \ge -n/2 = -\Theta(m/\ell)$.

It is an open question whether both the bounds on $\|H\|$ and
$\lambda_{\min}(H)$ can be saturated at the same time.
%
%


\subsection*{Acknowledgements}

We would like to thank Oded Regev for the key idea behind the proof of
Lemma \ref{lem:quantum1}. This work was initiated at the BIRS workshop
15w5098, ``Hypercontractivity and Log Sobolev Inequalities in Quantum
Information Theory''. We would like to thank BIRS and the Banff Centre
for their hospitality. AM was supported by the UK EPSRC under Early
Career Fellowship EP/L021005/1.  AWH was funded by NSF grants
CCF-1111382 and CCF-1452616 and ARO contract 
W911NF-12-1-0486.
No new data were created during this study.



\bibliographystyle{revtex4-1}
\bibliography{refs}

\end{document}